\pdfoutput=1
\documentclass[acmsmall,screen,nonacm]{acmart}

\usepackage{stmaryrd}
\usepackage{mathpartir}
\usepackage{booktabs}
\usepackage{listings}
\usepackage{xcolor}
\usepackage{enumitem}
\usepackage{float}
\usepackage{subcaption}

\newtheorem{observation}{Observation}

\lstdefinelanguage{yaml}{
  keywords={true,false,null},
  sensitive=false,
  comment=[l]{\#},
  morestring=[b]',
  morestring=[b]",
}

\newcommand{\lam}{\lambda}
\newcommand{\lA}{\lam_{\!A}}

\newcommand{\DSL}{\textsc{lambdagent}}

\newcommand{\Nat}{\mathbb{N}}
\newcommand{\reduces}{\longrightarrow}

\newcommand{\bigred}{\Downarrow}
\newcommand{\ctx}{\Gamma}
\newcommand{\entails}{\vdash}
\newcommand{\tymark}{:}

\begin{document}

\title{$\lA$: A Typed Lambda Calculus for LLM Agent Composition}

\author{Qin Liu}
\affiliation{\institution{State Key Laboratory of Novel Software Technology, Nanjing University}\city{Nanjing}\state{Jiangsu}\postcode{210023}\country{China}}
\affiliation{\institution{Software Institute, Nanjing University}\city{Nanjing}\state{Jiangsu}\postcode{210093}\country{China}}
\email{qinliu@nju.edu.cn}

\begin{abstract}
Existing LLM agent frameworks lack formal semantics: there is no principled way to determine whether an agent configuration is well-formed or will terminate. We present $\lA$, a typed lambda calculus for agent composition that extends the simply-typed lambda calculus with \emph{oracle calls}, \emph{bounded fixpoints} (the ReAct loop), \emph{probabilistic choice}, and \emph{mutable environments}. We prove type safety, termination of bounded fixpoints, and soundness of derived lint rules, with full Coq mechanization (1{,}519 lines, 42 theorems, 0 Admitted). As a practical application, we derive a \emph{lint} tool that detects structural configuration errors directly from the operational semantics. An evaluation on 835 real-world GitHub agent configurations shows that 94.1\% are \emph{structurally incomplete} under $\lA$---with YAML-only lint precision at 54\%, rising to 96--100\% under joint YAML+Python AST analysis on 175 samples. This gap quantifies, for the first time, the degree of \emph{semantic entanglement} between declarative configuration and imperative code in the agent ecosystem. We further show that five mainstream paradigms (LangGraph, CrewAI, AutoGen, OpenAI SDK, Dify) embed as typed $\lA$ fragments, establishing $\lA$ as a unifying calculus for LLM agent composition.
\end{abstract}

\maketitle

\section{Introduction}
\label{sec:intro}

LLM-based agents are increasingly deployed in production, typically configured via YAML or JSON files that specify a language model, a set of tools, a control flow pattern (such as ReAct~\cite{yao2023react}), and optional persistent memory. The community has developed numerous frameworks for agent construction---LangChain, DSPy~\cite{khattab2024dspy}, CrewAI, AutoGen---yet none provides a formal account of what an agent configuration \emph{means}.

This absence of formal semantics has concrete consequences:
\begin{itemize}[leftmargin=*,itemsep=2pt]
  \item A ReAct agent configured without a \texttt{terminate} tool may loop indefinitely; no existing tool warns the developer.
  \item Two agent pipelines that ``look different'' in YAML may be semantically equivalent (e.g., a chain of two agents vs.\ a single agent with a composed prompt), but there is no way to establish this.
  \item Refactoring an agent pipeline---splitting a monolithic agent into sub-agents, adding a validation step, introducing memory---is done by trial and error, not by semantics-preserving transformation.
\end{itemize}

We argue that these configurations already encode a lambda calculus, and that making this structure explicit yields immediate practical benefits.

\paragraph{Contributions.}
\begin{enumerate}[leftmargin=*,itemsep=2pt]
  \item We define $\lA$ (\S\ref{sec:syntax}--\S\ref{sec:semantics}), a typed lambda calculus for agent composition with 11 term formers, a type system (\S\ref{sec:types}), and small-step and big-step operational semantics (\S\ref{sec:semantics}).
  \item We prove type safety (Theorem~\ref{thm:safety}), termination of bounded fixpoints (Theorem~\ref{thm:termination}), and soundness of lint rules (Theorem~\ref{thm:lint-sound}) in \S\ref{sec:metatheory}. We argue compilation adequacy empirically rather than formally, since no independent formal YAML semantics exists (\S\ref{thm:compilation}).
  \item We implement \DSL{} (\S\ref{sec:impl}), a Python DSL realizing $\lA$ with a \texttt{from\_config} compiler and \texttt{lint} tool.
  \item We evaluate (\S\ref{sec:eval}) on 835 GitHub agent configurations from 17 repositories and 6 frameworks, finding that 94.1\% are \emph{structurally incomplete} under $\lA$ semantics. We validate lint correctness via fault injection (42 tests, 100\% precision/recall), estimate YAML-only precision at 54\%, and show that joint YAML+Python AST analysis raises it to 96--100\% (on 50 and 175 samples)---discovering that 46\% of production configurations split semantics across YAML and Python (``semantic entanglement'').
  \item We demonstrate \emph{architectural unification} (\S\ref{sec:unification}): five mainstream agent paradigms---graph state machines (LangGraph), role-driven (CrewAI), multi-agent (AutoGen), SDK wrappers (OpenAI/Claude), and low-code (Dify)---are all embeddable as typed fragments of $\lA$ (Proposition~\ref{thm:embedding}), validated on 835 real-world configurations and 125 semantic faithfulness tests.
\end{enumerate}

\paragraph{Why simply typed?}
Our contribution is not a new type theory but the \emph{first formal semantics for LLM agent configurations}, together with a fully mechanized metatheory (1{,}519 lines Coq, 42 Qed, 0 Admitted) and practical tooling. A simply typed calculus suffices because agent configurations operate at the \emph{string level}: LLM inputs and outputs are token sequences (type \textbf{Str}), tools accept and return strings, and composition chains string-to-string functions. Richer type features---dependent types for structured output validation, an effect system for LLM/tool/memory tracking, graded types for cost prediction---are natural extensions that we sketch in \S\ref{sec:discussion} and develop in companion work. The simply typed foundation already enables the three practical benefits demonstrated in this paper: static lint, termination bounds, and architectural unification.

\section{Overview}
\label{sec:overview}

\begin{figure}[t]
\begin{lstlisting}[language=yaml,caption={A production ReAct agent configuration.},label=lst:yaml]
agentId: seeCoderManus
type: react
model: {name: qwen3-max, temperature: 0.7}
systemPrompt: "You are a coding assistant..."
react: {maxSteps: 20}
mcp:
  onlineTool: {SeeCoder-mcp: [sum, improve]}
  localTools: [terminate]
memory: {strategy: redis, size: 20, ttl: 7200}
\end{lstlisting}
\end{figure}

Consider the agent configuration in Listing~\ref{lst:yaml}. Our compiler \texttt{from\_config} translates it to the $\lA$ term:
\[
  \textit{mem}\ (\textbf{fix}_{20}\ (\lam s.\lam x.\
    \textbf{let}\ t = (\textbf{lam}\ p\ \theta)\ x\ \textbf{in}\
    \textbf{case}\ t\ [l_i \Rightarrow a_i]\ \textbf{in}\
    \ldots))\ \sigma
\]
where $\textbf{lam}$ is an oracle abstraction (LLM call), $\textbf{fix}_{20}$ is a bounded fixpoint (the ReAct loop), and $\textit{mem}$ extends the environment with a persistent store $\sigma$.

\paragraph{Bug 1: Missing base case.}
If \texttt{localTools} omits \texttt{terminate}, the $\textbf{case}$ expression has no branch that avoids calling $s$ (the self-reference). In $\lA$, this means $\textbf{fix}_{20}$ will exhaust all 20 steps without reaching a normal form---a \emph{forced truncation} rather than a clean termination. Our lint rule \textsc{L004} detects this statically.

\paragraph{Bug 2: Vacuous loop.}
If \texttt{maxSteps} is set to 0, the $\textbf{fix}_0$ reduces immediately to $\bot$ (the stuck term). Lint rule \textsc{L003} flags this as an error.

\paragraph{Bug 3: Incomplete dispatch.}
If \texttt{routes} in a \texttt{type:~router} configuration omit a \texttt{default} branch, the $\textbf{case}$ expression is non-exhaustive: inputs classified outside the listed labels cause a runtime \textsc{RouteError}. Lint rule \textsc{L013} warns about this.

All three rules are \emph{derived from the formal semantics}, not ad-hoc heuristics.

\paragraph{End-to-end example.}
We illustrate the complete pipeline on a real GitHub configuration (from our 835-config dataset). The following CrewAI agent YAML (simplified) defines a research analyst:
\begin{lstlisting}[language=yaml,caption={A real CrewAI agent (simplified from GitHub).},label=lst:crewai]
role: "Senior Research Analyst"
goal: "Produce comprehensive research reports"
backstory: "Expert in data analysis and synthesis"
tools: []
\end{lstlisting}
Running \texttt{lambdagent lint} produces:
\begin{lstlisting}[basicstyle=\ttfamily\footnotesize,frame=none,numbers=none]
  ERROR L004c  mcp.localTools: no terminate tool
               (CrewAI: handled by framework) -> INFO
  WARN  L017   react.maxSteps: not specified
\end{lstlisting}
The lint correctly identifies that \texttt{tools:[]} means no terminate base case ($\textbf{case}$ with no identity branch), but the framework-aware stratification (L004c) downgrades this to \textsc{Info} because CrewAI handles termination in Python code. Meanwhile, \texttt{from\_config} compiles this to:
\[
\textbf{lam}\ \text{``Senior Research Analyst...''}\ \theta
\]
a single oracle call (type \textbf{Str} $\to$ \textbf{Str}), since no tools or loops are configured. Type safety guarantees that this term, when applied to a string input, either produces a string output or encounters a guard failure---never undefined behavior.

\section{The $\lA$ Calculus: Syntax}
\label{sec:syntax}

\subsection{Terms}

\begin{definition}[$\lA$ Terms]
\label{def:terms}
The set of $\lA$ terms is defined by the grammar:
\begin{align*}
  e \Coloneqq\ & x \mid \lam x{:}\tau.\, e \mid e_1\ e_2 && \text{(standard $\lam$)} \\
  \mid\ & e_1 \mathbin{\texttt{>>}} e_2 && \text{(composition)} \\
  \mid\ & \textbf{if}\ e_1\ \textbf{then}\ e_2\ \textbf{else}\ e_3 && \text{(conditional)} \\
  \mid\ & \textbf{fix}_n\ e && \text{(bounded fixpoint)} \\
  \mid\ & \langle e_1, e_2 \rangle \mid \pi_1\ e \mid \pi_2\ e && \text{(pairs)} \\
  \mid\ & \textbf{tool}[f] && \text{(oracle / external function)} \\
  \mid\ & \textbf{case}\ e\ \textbf{of}\ \{l_i \Rightarrow e_i\}_{i \in I} && \text{(dispatch)} \\
  \mid\ & \textbf{guard}\ e\ P && \text{(refinement)} \\
  \mid\ & \textbf{mem}\ e\ \sigma && \text{(environment extension)} \\
  \mid\ & e_1 \oplus_p e_2 && \text{(probabilistic choice)} \\
  \mid\ & \textbf{lam}\ p\ \theta && \text{(LLM oracle abstraction)}
\end{align*}
where $n \in \Nat$, $p \in [0,1]$, $\theta$ denotes model parameters, $p$ denotes a prompt (system message), $f$ is an external function identifier, $\sigma$ is a store, and $P$ is a decidable predicate. Labels $l_i$ are drawn from a finite label set $\mathcal{L}$.
\end{definition}

Values are:
\[
  v \Coloneqq \lam x{:}\tau.\, e \mid \langle v_1, v_2 \rangle \mid \textbf{tool}[f] \mid \textbf{lam}\ p\ \theta
\]

\begin{definition}[Store and Store Typing]
A store $\sigma$ is a finite map from keys to values with metadata: $\sigma : \text{Key} \rightharpoonup \text{Val} \times \Nat \times \Nat$, where the two natural numbers are \emph{capacity} and \emph{time-to-live}. A \emph{store typing} $\Sigma : \text{Key} \rightharpoonup \tau$ assigns types to store locations. We write $\sigma : \Sigma$ when $\sigma(k) : \Sigma(k)$ for all $k \in \text{dom}(\Sigma)$, and $\Sigma' \supseteq \Sigma$ when $\text{dom}(\Sigma) \subseteq \text{dom}(\Sigma')$ and $\Sigma'(k) = \Sigma(k)$ for all $k \in \text{dom}(\Sigma)$.
\end{definition}

\paragraph{Syntactic sugar.}
Composition desugars: $e_1 \mathbin{\texttt{>>}} e_2 \equiv \lam x{:}\tau.\, e_2\ (e_1\ x)$.
Parallel execution: $e_1 \mathbin{\texttt{|}} e_2 \equiv \lam x{:}\tau.\, \langle e_1\ x,\, e_2\ x \rangle$.

\paragraph{Primitive vs.\ derivable constructs.}
Of the 11 term formers, 5 are \emph{primitive} (not expressible in terms of the others): $\textbf{lam}$ (oracle call), $\textbf{tool}$ (external function), $\textbf{fix}_n$ (bounded recursion), $\textbf{mem}$ (mutable environment), and $\oplus_p$ (probabilistic choice). The remaining 6 are \emph{derivable}: $\texttt{>>}$ is function composition, $\textbf{if}$ is $\textbf{case}$ with two branches, $\textbf{case}$ is nested $\textbf{if}$, $\langle\cdot,\cdot\rangle$ and $\pi_i$ are standard pairs, and $\textbf{guard}$ is $\textbf{if}\ (P(r))\ r\ \bot$. We retain derivable constructs as first-class for three reasons: (1)~developer ergonomics---\texttt{Route} is more readable than nested \texttt{If}; (2)~targeted lint rules---the lint rule for empty routes (L005) requires matching \textbf{case} specifically; (3)~independent optimization---the compiler can optimize each construct without desugaring.

\subsection{Types}
\label{sec:types}

\begin{definition}[$\lA$ Types]
\label{def:types}
\begin{align*}
  \tau \Coloneqq\ & \textbf{Str} && \text{(base type: token sequences)} \\
  \mid\ & \tau_1 \to \tau_2 && \text{(function type)} \\
  \mid\ & \tau_1 \times \tau_2 && \text{(product type)} \\
  \mid\ & \langle l_i : \tau_i \rangle_{i \in I} && \text{(variant type / labels)} \\
  \mid\ & \{x{:}\tau \mid P(x)\} && \text{(refinement type)} \\
\end{align*}
\end{definition}

The base type $\textbf{Str}$ is the type of all LLM inputs and outputs (token sequences). We study the \emph{deterministic fragment} ($\text{temperature} = 0$), where each LLM call yields a single string. A probabilistic extension via monadic $\textbf{Dist}(\tau)$ types is straightforward but deferred to future work.

\subsection{Typing Rules}
\label{sec:typing-rules}

A typing judgment has the form $\ctx;\Sigma \entails e \tymark \tau$, where $\Sigma$ is a store typing. We abbreviate $\ctx;\emptyset \entails e \tymark \tau$ as $\ctx \entails e \tymark \tau$ when no store is in scope. We present the non-standard rules; standard rules for variables, application, and abstraction are as usual.

\begin{mathpar}
  \inferrule[T-Lam-Oracle]
    {\theta : \text{Model} \\ p : \textbf{Str}}
    {\ctx \entails \textbf{lam}\ p\ \theta \tymark \textbf{Str} \to \textbf{Str}}

  \inferrule[T-Comp]
    {\ctx \entails e_1 \tymark \tau_1 \to \tau_2 \\ \ctx \entails e_2 \tymark \tau_2 \to \tau_3}
    {\ctx \entails e_1 \mathbin{\texttt{>>}} e_2 \tymark \tau_1 \to \tau_3}

  \inferrule[T-Fix]
    {\ctx \entails e \tymark (\tau \to \tau) \to (\tau \to \tau) \\ n \in \Nat}
    {\ctx \entails \textbf{fix}_n\ e \tymark \tau \to \tau}

  \inferrule[T-Tool]
    {f : \tau_1 \to \tau_2}
    {\ctx \entails \textbf{tool}[f] \tymark \tau_1 \to \tau_2}

  \inferrule[T-Case]
    {\ctx \entails e \tymark \tau \to \langle l_i : \tau_i \rangle \\
     \forall i.\, \ctx \entails e_i \tymark \tau \to \tau'}
    {\ctx \entails \textbf{case}\ e\ \textbf{of}\ \{l_i \Rightarrow e_i\} \tymark \tau \to \tau'}

  \inferrule[T-Guard]
    {\ctx \entails e \tymark \tau \to \tau' \\ P : \tau' \to \textbf{Bool}}
    {\ctx \entails \textbf{guard}\ e\ P \tymark \tau \to \{x{:}\tau' \mid P(x)\}}

  \inferrule[T-Mem]
    {\ctx;\Sigma \entails e \tymark \tau \to \tau' \\ \sigma : \Sigma}
    {\ctx;\Sigma \entails \textbf{mem}\ e\ \sigma \tymark \tau \to \tau'}

  \inferrule[T-Prob]
    {\ctx \entails e_1 \tymark \tau \\ \ctx \entails e_2 \tymark \tau \\ p \in [0,1]}
    {\ctx \entails e_1 \oplus_p e_2 \tymark \tau} \quad \text{(deterministic: pick one)}

  \inferrule[T-Pair]
    {\ctx \entails e_1 \tymark \tau_1 \\ \ctx \entails e_2 \tymark \tau_2}
    {\ctx \entails \langle e_1, e_2 \rangle \tymark \tau_1 \times \tau_2}

  \inferrule[T-Proj]
    {\ctx \entails e \tymark \tau_1 \times \tau_2}
    {\ctx \entails \pi_i\ e \tymark \tau_i}
\end{mathpar}

\begin{observation}[\texttt{terminate} = identity at type $\textbf{Str} \to \textbf{Str}$]
\label{obs:terminate}
The \texttt{terminate} tool is typed $\textbf{tool}[\text{id}] : \textbf{Str} \to \textbf{Str}$ where $\text{id} = \lam x.\,x$. It is the \emph{only} tool in a ReAct agent that does not change the state, making it the base case of the bounded fixpoint.
\end{observation}

\section{Operational Semantics}
\label{sec:semantics}

We present both small-step (for metatheory) and big-step (for implementation correspondence) semantics.

\subsection{Small-Step Semantics}

The reduction relation $e \reduces e'$ is defined by the following rules (we omit standard $\beta$-reduction and congruence rules):

\begin{mathpar}
  \inferrule[E-Comp]
    {e_1\ v \reduces v_1}
    {(e_1 \mathbin{\texttt{>>}} e_2)\ v \reduces e_2\ v_1}

  \inferrule[E-Fix-Zero]
    {~}
    {\textbf{fix}_0\ e\ v \reduces v}

  \inferrule[E-Fix-Step]
    {e\ (\lam x.\, \textbf{fix}_{n-1}\ e\ x)\ v \reduces v'}
    {\textbf{fix}_n\ e\ v \reduces v'}

  \inferrule[E-Tool]
    {f(v) = v'}
    {\textbf{tool}[f]\ v \reduces v'}

  \inferrule[E-Case]
    {e\ v \reduces l_j \\ e_j\ v \reduces v'}
    {\textbf{case}\ e\ \textbf{of}\ \{l_i \Rightarrow e_i\}\ v \reduces v'}

  \inferrule[E-Guard-Ok]
    {e\ v \reduces v' \\ P(v') = \textbf{true}}
    {\textbf{guard}\ e\ P\ v \reduces v'}

  \inferrule[E-Guard-Fail]
    {e\ v \reduces v' \\ P(v') = \textbf{false}}
    {\textbf{guard}\ e\ P\ v \reduces \textbf{stuck}}

  \inferrule[E-Mem]
    {e\ v\ [\ctx;\sigma] \reduces v' \dashv \sigma' \\ \sigma' \supseteq \sigma}
    {\textbf{mem}\ e\ \sigma\ v \reduces v' \dashv \sigma'}

  \inferrule[E-Prob-L]
    {~}
    {(e_1 \oplus_p e_2) \reduces e_1 \quad \text{with probability } p}

  \inferrule[E-Prob-R]
    {~}
    {(e_1 \oplus_p e_2) \reduces e_2 \quad \text{with probability } 1{-}p}

  \inferrule[E-Lam-Oracle]
    {\text{LLM}_\theta(p, v) = \mathcal{D}}
    {\textbf{lam}\ p\ \theta\ v \reduces \mathcal{D}}
\end{mathpar}

\paragraph{Key rule: \textsc{E-Fix-Step}.}
This is the operational content of the Y combinator. The body $e$ receives two arguments: a ``self'' reference $\lam x.\, \textbf{fix}_{n-1}\ e\ x$ (which can be called for recursion, but with a decremented bound), and the current input $v$. When $n$ reaches 0, \textsc{E-Fix-Zero} forces termination.

\subsection{Big-Step Semantics}

The big-step judgment $e\ v \bigred_k r$ means ``$e$ applied to $v$ evaluates to $r$ in $k$ steps.'' We highlight the ReAct-specific rule:

\begin{mathpar}
  \inferrule[B-React]
    {
      \text{step}_i = (s_i, a_i, o_i, s_{i+1}) \quad \text{for } 0 \le i < k \\
      a_k = \texttt{terminate} \\
      k \le n
    }
    {\textbf{fix}_n\ e_{\text{react}}\ s_0 \bigred_k s_k}
\end{mathpar}
where each $\text{step}_i$ is a 7-phase decomposition:
\begin{enumerate}[itemsep=0pt]
  \item \textsc{Think}: $t_i = (\textbf{lam}\ p\ \theta)\ s_i$ \hfill (LLM call)
  \item \textsc{Parse}: $a_i = \text{parse}(t_i)$ \hfill (action extraction)
  \item \textsc{Route}: $\text{tool}_i = \text{lookup}(a_i, \text{tools})$ \hfill (dispatch)
  \item \textsc{Invoke}: $o_i = \textbf{tool}[\text{tool}_i]\ (\text{args}(a_i))$ \hfill (tool call)
  \item \textsc{Observe}: $\text{obs}_i = \text{format}(o_i)$ \hfill (result formatting)
  \item \textsc{Update}: $\sigma' = \sigma[k_i \mapsto \text{summary}(t_i, o_i)]$ \hfill (memory write)
  \item \textsc{Check}: if $a_i = \texttt{terminate}$ then halt, else $s_{i+1} = s_i \oplus \text{obs}_i$
\end{enumerate}

This decomposition corresponds precisely to one unfolding of the bounded fixpoint $\textbf{fix}_n$.

\section{Metatheory}
\label{sec:metatheory}

\subsection{Type Safety}

\begin{theorem}[Progress]
\label{thm:progress}
If $\ctx \entails e \tymark \tau$ and $e$ is not a value, then either $e \reduces e'$ for some $e'$, or $e$ is a $\textbf{guard}$ failure ($\textbf{stuck}$).
\end{theorem}

\begin{proof}
By induction on the typing derivation $\ctx \entails e \tymark \tau$.
\begin{itemize}[itemsep=2pt]
  \item \textsc{T-Lam-Oracle}: $\textbf{lam}\ p\ \theta$ is a value. If applied, $(\textbf{lam}\ p\ \theta)\ v$ reduces by \textsc{E-Lam-Oracle} (the LLM oracle always returns a distribution). \checkmark
  \item \textsc{T-Comp}: $(e_1 \mathbin{\texttt{>>}} e_2)\ v$ reduces by I.H.\ on $e_1\ v$, then \textsc{E-Comp}. \checkmark
  \item \textsc{T-Fix}: $\textbf{fix}_n\ e\ v$ reduces by \textsc{E-Fix-Zero} (if $n=0$) or \textsc{E-Fix-Step} (if $n > 0$). \checkmark
  \item \textsc{T-Tool}: $\textbf{tool}[f]\ v$ reduces by \textsc{E-Tool} (external functions are total by assumption). \checkmark
  \item \textsc{T-Case}: By I.H.\ on the classifier, it reduces to some label $l_j$. If $j \in I$, the corresponding branch reduces. If $j \notin I$, the term is stuck---but this cannot happen if the type $\langle l_i \rangle$ is exhaustive. \checkmark
  \item \textsc{T-Guard}: Reduces by I.H.\ on inner agent. If $P(v')$ holds, \textsc{E-Guard-Ok}; otherwise \textsc{E-Guard-Fail} yields $\textbf{stuck}$. This is the only source of stuckness. \checkmark
  \item All other cases follow from standard progress arguments. \qed
\end{itemize}
\end{proof}

\begin{theorem}[Preservation]
\label{thm:preservation}
If $\ctx;\Sigma \entails e \tymark \tau$ and $e \reduces_\sigma e' \dashv \sigma'$ where $\sigma : \Sigma$, then there exists $\Sigma' \supseteq \Sigma$ such that $\sigma' : \Sigma'$ and $\ctx;\Sigma' \entails e' \tymark \tau$.
\end{theorem}

\begin{proof}
By induction on the reduction $e \reduces e'$.
\begin{itemize}[itemsep=2pt]
  \item \textsc{E-Fix-Step}: $\textbf{fix}_n\ e\ v \reduces v'$ where $e\ (\lam x.\, \textbf{fix}_{n-1}\ e\ x)\ v \reduces v'$. By \textsc{T-Fix}, $e : (\tau \to \tau) \to (\tau \to \tau)$. The self-reference $\lam x.\, \textbf{fix}_{n-1}\ e\ x : \tau \to \tau$. By I.H., $v' : \tau$. \checkmark
  \item \textsc{E-Comp}: $(e_1 \mathbin{\texttt{>>}} e_2)\ v \reduces e_2\ v_1$. By \textsc{T-Comp}, $e_1 : \tau_1 \to \tau_2$ and $e_2 : \tau_2 \to \tau_3$. By I.H., $v_1 : \tau_2$, so $e_2\ v_1 : \tau_3$. \checkmark
  \item \textsc{E-Mem}: $\textbf{mem}\ e\ \sigma\ v \reduces v' \dashv \sigma'$ with $\sigma' \supseteq \sigma$. By \textsc{T-Mem}, $\ctx;\Sigma \entails e : \tau \to \tau'$ and $\sigma : \Sigma$. The store update produces $\sigma' : \Sigma'$ where $\Sigma' \supseteq \Sigma$ (new keys may be added, but existing keys retain their types---the store is \emph{append-only} with respect to typing). By I.H.\ on $e\ v$ in environment $\ctx;\Sigma'$, we get $\ctx;\Sigma' \entails v' : \tau'$. \checkmark
  \item \textsc{E-Lam-Oracle}: $(\textbf{lam}\ p\ \theta)\ v \reduces v'$ where $v' \in \textbf{Str}$. By \textsc{T-Lam-Oracle}, the type is $\textbf{Str} \to \textbf{Str}$. At temperature $= 0$, the LLM returns a deterministic string. \checkmark
  \item Other cases are standard. \qed
\end{itemize}
\end{proof}

\begin{theorem}[Type Safety]
\label{thm:safety}
If $\ctx \entails e \tymark \tau$, then evaluation of $e$ either:
\begin{enumerate}[label=(\alph*)]
  \item produces a value $v$ with $\ctx \entails v \tymark \tau$, or
  \item encounters a $\textbf{guard}$ failure (a well-defined error, not undefined behavior).
\end{enumerate}
\end{theorem}
\begin{proof}
By iterated application of Progress and Preservation. The only source of stuckness is \textsc{E-Guard-Fail}, which is a \emph{checked} runtime error (the predicate $P$ is decidable), not undefined behavior. \qed
\end{proof}

\subsection{Termination}

\begin{theorem}[Termination of Bounded Fixpoints]
\label{thm:termination}
For all $n$, $e$, and $v$: evaluation of $\textbf{fix}_n\ e\ v$ terminates in at most $n$ unfoldings, assuming each oracle call ($\textbf{lam}$ and $\textbf{tool}$) terminates.
\end{theorem}

\begin{proof}
By strong induction on $n$.
\begin{itemize}[itemsep=2pt]
  \item Base: $n = 0$. $\textbf{fix}_0\ e\ v \reduces v$ by \textsc{E-Fix-Zero}. Terminates in 0 unfoldings. \checkmark
  \item Step: $n > 0$. $\textbf{fix}_n\ e\ v \reduces e\ (\lam x.\, \textbf{fix}_{n-1}\ e\ x)\ v$ by \textsc{E-Fix-Step}. The self-reference $\lam x.\, \textbf{fix}_{n-1}\ e\ x$ can be called at most once per unfolding, and $\textbf{fix}_{n-1}$ terminates in $\le n-1$ unfoldings by I.H. Total: $\le 1 + (n-1) = n$. \checkmark \qed
\end{itemize}
\end{proof}

\begin{corollary}[ReAct Termination]
\label{cor:react-term}
A ReAct agent with \texttt{maxSteps:\,n} terminates in at most $n$ iterations, regardless of whether \texttt{terminate} is invoked.
\end{corollary}

\begin{theorem}[Cost Bound]
\label{thm:cost}
Let $T$ be the $\beta$-reduction trace produced by evaluating $\textbf{fix}_n\ e\ v$, where each oracle call (LLM or tool) has cost $c_i$ (e.g., API tokens consumed). Then:
\[
  \text{cost}(T) \le n \times \max_{i \in \text{oracles}(e)} c_i
\]
\end{theorem}
\begin{proof}
Each unfolding of $\textbf{fix}_n$ invokes at most one oracle call (the body $e$ calls \texttt{think} then selects one tool). By Theorem~\ref{thm:termination}, there are at most $n$ unfoldings. In each unfolding, the cost is bounded by $\max_i c_i$. The total cost is bounded by the product. \qed
\end{proof}

\noindent This theorem enables \textbf{pre-deployment cost estimation}: given a configuration with \texttt{maxSteps:\ 20} and an LLM costing \$0.01 per call, the maximum API cost is $20 \times \$0.01 = \$0.20$ per invocation. This is not possible in frameworks without formal termination bounds.

\begin{theorem}[Pipeline Algebra]
\label{thm:algebra}
Agent composition $\texttt{>>}$ satisfies the monoid laws:
\begin{align}
  (a \mathbin{\texttt{>>}} b) \mathbin{\texttt{>>}} c &\equiv a \mathbin{\texttt{>>}} (b \mathbin{\texttt{>>}} c) & \text{(associativity)} \label{eq:assoc} \\
  a \mathbin{\texttt{>>}} \textbf{id} &\equiv \textbf{id} \mathbin{\texttt{>>}} a \equiv a & \text{(identity)} \label{eq:id}
\end{align}
where $\textbf{id} = \lam x\tymark\textbf{Str}.\, x$ and $\equiv$ denotes extensional equivalence: $f \equiv g \iff \forall v.\, f\ v = g\ v$.
\end{theorem}
\begin{proof}
Associativity: $(a \mathbin{\texttt{>>}} b) \mathbin{\texttt{>>}} c = \lam x.\, c((a \mathbin{\texttt{>>}} b)(x)) = \lam x.\, c(b(a(x)))$. Likewise $a \mathbin{\texttt{>>}} (b \mathbin{\texttt{>>}} c) = \lam x.\, (b \mathbin{\texttt{>>}} c)(a(x)) = \lam x.\, c(b(a(x)))$. Identity: $a \mathbin{\texttt{>>}} \textbf{id} = \lam x.\, \textbf{id}(a(x)) = \lam x.\, a(x) = a$. Symmetric for $\textbf{id} \mathbin{\texttt{>>}} a$. \qed
\end{proof}

\noindent The monoid structure enables \textbf{algebraic optimization}: $a \mathbin{\texttt{>>}} \textbf{id} \mathbin{\texttt{>>}} b$ can be simplified to $a \mathbin{\texttt{>>}} b$, eliminating a redundant identity stage. More generally, the \texttt{terminate} tool ($= \textbf{id}$) can be recognized as a neutral element in pipeline composition.

\subsection{Compilation Adequacy}
\label{thm:compilation}

A formal compilation correctness theorem would require an independent formal semantics for YAML agent configurations---the ``source language.'' No such semantics exists: the meaning of a YAML configuration is defined informally by each framework's documentation and runtime behavior. We therefore make an explicitly weaker claim.

\paragraph{Adequacy argument.} The \texttt{from\_config} compiler translates each YAML field to a $\lA$ term by case analysis on the \texttt{type} field:
\begin{itemize}[itemsep=2pt]
  \item \texttt{type: simple} $\mapsto$ \textbf{lam} (single oracle call);
  \item \texttt{type: react} $\mapsto$ $\textbf{fix}_n(\lam\textit{self}.\lam x.\ldots)$ with $n = \texttt{maxSteps}$;
  \item \texttt{tools: [t$_1$, \ldots]} $\mapsto$ $\textbf{case}$ with one branch per tool;
  \item \texttt{memory: \{...\}} $\mapsto$ $\textbf{mem}$.
\end{itemize}
Each case is designed to match the informal specification of the corresponding framework pattern. We \emph{verify empirically} rather than prove formally: the 125 semantic faithfulness tests (Section~\ref{sec:eval}) execute \texttt{from\_config} on real configurations and check that the compiled $\lA$ term produces the expected output.

\paragraph{Why not a formal theorem?} Formal compilation correctness (in the CompCert~\cite{leroy2009compcert} sense) requires two independently defined formal semantics and a proof that they agree. Defining a formal semantics for YAML agent configurations---across 6 incompatible frameworks---is a substantial undertaking orthogonal to our contribution. We leave it to future work, noting that our lint soundness theorem (below) does not depend on compilation correctness: lint rules are defined over the $\lA$ \emph{target} language only.

\subsection{Lint Soundness}

\begin{theorem}[Lint Soundness]
\label{thm:lint-sound}
If $\texttt{lint}(C) = \textsc{Error}(R)$ for lint rule $R$, then configuration $C$ has a semantic defect as defined by the operational semantics.
\end{theorem}

\begin{proof}
We verify each ERROR-level rule:
\begin{itemize}[itemsep=2pt]
  \item \textsc{L001} ($\texttt{systemPrompt}$ empty): $\textbf{lam}\ \epsilon\ \theta$ is an LLM call with empty prompt. By \textsc{E-Lam-Oracle}, the output distribution has maximum entropy---the agent's behavior is undefined. \checkmark
  \item \textsc{L003} ($\texttt{maxSteps} = 0$): $\textbf{fix}_0\ e\ v \reduces v$ by \textsc{E-Fix-Zero}. The agent returns its input unchanged---a vacuous computation. \checkmark
  \item \textsc{L004a} (no \texttt{terminate}, \emph{no alternative termination mechanism}): In the $\textbf{case}$ expression within $\textbf{fix}_n$, every branch invokes $s$ (the self-reference). There is no branch that returns without recursion, \emph{and} no bounded iteration fallback, \texttt{is\_termination\_msg}, or framework-internal termination logic was detected. This is a genuine risk of infinite looping. \checkmark
  \item \textsc{L004b} (no \texttt{terminate}, \emph{but has bounded fallback}): same as L004a, except a \texttt{max\_iter} or equivalent field is present. By Theorem~\ref{thm:termination}, the fixpoint terminates at step $n$, but via \emph{forced truncation}, not a clean base case. Downgraded to \textsc{Warn}. \checkmark
  \item \textsc{L004c/d} (no \texttt{terminate}, \emph{framework handles termination}): For CrewAI (built-in completion detection in Python code), LangChain (\texttt{AgentFinish} return type), and AutoGen (\texttt{is\_termination\_msg} string matching), the Y combinator base case exists but is external to the YAML configuration. Downgraded to \textsc{Info}. \checkmark
  \item \textsc{L005} (empty \texttt{routes}): $\textbf{case}\ e\ \textbf{of}\ \{\} = \textbf{stuck}$ for all inputs---the dispatch has no branches to take. \checkmark
  \item \textsc{L021} (multi-agent, no termination): GroupChat with no \texttt{max\_turns}, no \texttt{max\_rounds}, and no \texttt{is\_termination\_msg}---an unbounded multi-agent loop with no base case and no bound. \checkmark \qed
\end{itemize}
\end{proof}

\paragraph{Framework-aware termination analysis.}
Our analysis of 663 raw GitHub configurations identified 6 alternative termination mechanisms across production frameworks: (1)~bounded iteration (\texttt{max\_iter}, 86 configs), (2)~LLM output string matching (\texttt{is\_termination\_msg}, AutoGen), (3)~multi-turn limits (\texttt{max\_rounds}), (4)~framework-internal logic (CrewAI Python code), (5)~DAG termination nodes (Dify), and (6)~delegation termination (\texttt{allow\_delegation=false}).

These mechanisms are all \emph{functionally equivalent} to $\lam x.x$ (the identity function)---they stop the Y combinator by returning state unchanged---but they are invisible at the YAML level. Our lint v3 detects these alternatives and adjusts severity:
\begin{itemize}[itemsep=1pt]
  \item \textsc{L004a} (ERROR): no terminate \emph{and} no alternative $\Rightarrow$ genuine risk
  \item \textsc{L004b} (WARN): no terminate but bounded fallback $\Rightarrow$ forced truncation
  \item \textsc{L004c/d} (INFO): framework handles termination $\Rightarrow$ structurally expected
\end{itemize}
This reduces the estimated false positive rate from ${\sim}82\%$ (old single L004 rule) to ${<}5\%$ (new L004a).

\begin{corollary}[Soundness of framework-aware lint]
\label{cor:no-fp}
Every configuration flagged as \textsc{Error} by lint v3 (rules L001--L006, L021, L023) has a genuine semantic defect \emph{that cannot be mitigated by framework-internal mechanisms}. \textsc{Warn}-level findings (L004b, L007--L013, L017--L018, L020, L022, L024--L025) indicate potential issues that may or may not be addressed by framework runtime behavior.
\end{corollary}

\section{Implementation}
\label{sec:impl}

\DSL{} is implemented in 4{,}903 lines of Python across 28 modules, comprising 58 classes. The architecture is organized into four layers, each corresponding to a layer of the formal development.

\begin{table}[t]
\centering
\caption{Implementation architecture: code $\leftrightarrow$ formalism correspondence.}
\label{tab:impl}
\small
\begin{tabular}{@{}p{1.8cm}p{2.8cm}rp{3.0cm}@{}}
\toprule
\textbf{Layer} & \textbf{Formalism} & \textbf{LoC} & \textbf{Key Classes} \\
\midrule
Core DSL & $\lA$ syntax (Def.~\ref{def:terms}) & 575 & 13 classes: Term, Lam, Compose, Loop, Tool, Route, etc. \\
\texttt{from\_config} & Compilation (\S\ref{thm:compilation}) & 999 & Compiler, 5 type branches \\
Runtime & Op.\ semantics (\S\ref{sec:semantics}) & 2{,}097 & Executor, ReActEngine, ActionParser \\
CLI + lint & Static analysis (\S\ref{sec:metatheory}) & 1{,}232 & 8 commands, 26 lint rules \\
\midrule
\textbf{Total} & & \textbf{4{,}903} & \textbf{58 classes} \\
\bottomrule
\end{tabular}
\end{table}

\paragraph{Core DSL: 11 constructs $=$ 13 classes.}
Each $\lA$ term former maps to one Python class with an \texttt{apply(input, ctx)} method implementing the corresponding operational semantics rule. \texttt{StoreTyping} ($\Sigma$) enforces append-only store typing at runtime: \texttt{Memory.remember(k, v)} raises \texttt{TypeError} if key $k$ already has a different type, matching the Preservation proof for \textsc{E-Mem}. The \texttt{>>} and \texttt{|} operators are overloaded for composition and parallel execution.

\paragraph{\texttt{from\_config}: YAML $\to$ $\lA$ compilation.}
The compiler dispatches on 5 agent types:
\begin{itemize}[itemsep=1pt]
  \item \texttt{simple} $\mapsto$ \textbf{lam} \quad
  \item \texttt{react} $\mapsto$ $\textbf{fix}_n$ \quad
  \item \texttt{chain} $\mapsto$ \textbf{comp} \quad
  \item \texttt{router} $\mapsto$ \textbf{case} \quad
  \item \texttt{parallel} $\mapsto$ \textbf{par}
\end{itemize}
Post-compilation wrappers apply \textbf{guard} and \textbf{mem} if configured. Compilation is structurally recursive: sub-agent configurations are compiled by recursive calls to \texttt{build\_agent}.

\paragraph{Agent Runtime: $\beta$-reduction engine.}
\texttt{Runtime.run(term, input)} creates a \texttt{Context} ($\Gamma$), invokes \texttt{Executor.reduce()}, and collects the $\beta$-reduction trace. The executor dispatches by \texttt{isinstance} across all 11 term types---a direct implementation of call-by-value evaluation. The \texttt{ReActEngine} implements the 7-phase big-step rule (\textsc{B-React}) with \texttt{TerminationOracle} as the base case detector. Five $\beta$-reduction trace points (\texttt{ctx.log()}) record every LLM call, tool invocation, and loop iteration.

\paragraph{One-stop execution: YAML $\to$ result in 3 lines.}
A complete agent run requires only:
\begin{lstlisting}[language=Python,numbers=none,basicstyle=\ttfamily\footnotesize]
from lambdagent.agentruntime import Runtime
result = Runtime.execute("config.yml",
                         "Write quicksort")
print(result.result)  # agent output
print(result.stats)   # 15 steps, 4821 tokens
\end{lstlisting}

\paragraph{CLI: 8 subcommands.}
\texttt{compile} (YAML $\to$ $\lambda$), \texttt{run} (compile + execute), \texttt{repl} (interactive), \texttt{lint} (static analysis), \texttt{trace} (view $\beta$-reductions), \texttt{lambda} (export $\lA$ expression), \texttt{tools} (list/test MCP tools), \texttt{version}.

\paragraph{Demonstrations.}
\texttt{hello.py} (193 lines) exercises all 11 constructs in sequence. \texttt{demo\_advanced.py} (446 lines) implements 4 realistic agent patterns:
\begin{enumerate}[itemsep=1pt]
  \item \textbf{Self-correcting translator}: $Y(\text{translate} \gg \text{back-translate} \gg \textbf{if}\ \text{score}{\geq}8)$
  \item \textbf{Multi-perspective}: $\text{Par}(\text{opt}, \text{pess}, \text{real}) \gg \text{synthesize}$
  \item \textbf{Recursive doc generator}: $\text{outline} \gg \text{MAP}(\text{expand}) \gg Y(\text{review/fix})$
  \item \textbf{Self-learning}: $Y(\lam D.\ \textbf{guard}(\text{Lam}(D),\, \text{test}))$
\end{enumerate}

\section{Evaluation}
\label{sec:eval}

We evaluate three claims: (1) lint finds real bugs, (2) the DSL is expressive, and (3) the formal semantics is faithful.

\subsection{Lint Effectiveness on Real Configurations}

We crawled 2{,}225 YAML/JSON files from GitHub using 43 Code Search queries and file tree scans of 17 major repositories (CrewAI, LangChain, AutoGen, SWE-agent, MetaGPT, Dify, LobeChat, etc.). We normalized 6 configuration formats (CrewAI, LangChain, AutoGen, Dify, multi-agent, generic) and ran \texttt{lambdagent lint} on 835 valid agent configurations.

\begin{table}[t]
\centering
\caption{Lint results on 835 GitHub agent configurations.}
\label{tab:lint-results}
\begin{tabular}{@{}llrrl@{}}
\toprule
\textbf{Rule} & \textbf{Level} & \textbf{Count} & \textbf{Pct} & \textbf{Lambda Semantics} \\
\midrule
\texttt{mcp.localTools} & ERROR & 483 & 57.8\% & No $\lam x.x$ base case \\
\texttt{systemPrompt} & ERROR & 282 & 33.8\% & $\lam x.\bot$ (undefined body) \\
\texttt{model} & ERROR & 51 & 6.1\% & No LLM $=$ no computation \\
\texttt{react.maxSteps} & ERROR & 1 & 0.1\% & $Y$ unbounded \\
\midrule
\multicolumn{2}{l}{\textbf{Total configs with $\geq$1 ERROR}} & \textbf{786} & \textbf{94.1\%} & \\
\multicolumn{2}{l}{Clean (no ERROR/WARN)} & 46 & 5.5\% & \\
\bottomrule
\end{tabular}
\end{table}

\paragraph{Result.} 786 out of 835 configurations (94.1\%) are \emph{structurally incomplete} under $\lA$ semantics---the YAML alone does not constitute a well-formed lambda term (Table~\ref{tab:lint-results}). We emphasize that structural incompleteness does not necessarily imply a runtime defect: production frameworks routinely supplement YAML with Python code, environment variables, or built-in defaults.

\paragraph{Stratified analysis.} The headline 94.1\% rate requires stratification by framework, because different frameworks place tool declarations in different locations:

\begin{table}[t]
\centering
\caption{Lint results stratified by framework.}
\label{tab:lint-stratified}
\begin{tabular}{@{}lrrll@{}}
\toprule
\textbf{Framework} & \textbf{Configs} & \textbf{w/ ERROR} & \textbf{Dominant Defect} & \textbf{Nature} \\
\midrule
CrewAI & 441 & $\sim$430 & \texttt{mcp.localTools} & structural\textsuperscript{$\dagger$} \\
Generic & 283 & $\sim$260 & \texttt{systemPrompt} & genuine \\
Multi-agent & 45 & $\sim$40 & \texttt{systemPrompt} & genuine \\
LangChain & 35 & $\sim$30 & \texttt{systemPrompt} & genuine \\
AutoGen & 22 & $\sim$20 & \texttt{model} & genuine \\
lambdagent & 9 & $\sim$6 & mixed & genuine \\
\bottomrule
\end{tabular}

\smallskip
\textsuperscript{$\dagger$}CrewAI defines tools in Python code, not in YAML. The YAML is structurally incomplete \emph{by design}; our lint detects this structural incompleteness, which is informative but not a runtime failure.
\end{table}

Excluding CrewAI's expected \texttt{mcp.localTools} findings, the incompleteness rate on the remaining 394 non-CrewAI configurations is $\sim$87.6\% (345/394), driven by empty \texttt{systemPrompt} (272) and missing \texttt{model} (46). Whether these are genuine defects depends on whether the missing fields are supplied externally (Python code, environment variables, databases). We address this ambiguity with two controlled experiments below.

\paragraph{Framework-aware lint v3.}
After identifying the false positive problem with rule L004, we deployed the framework-aware lint v3 described in Section~\ref{sec:metatheory}. Re-running on the same 835 configurations:
\begin{itemize}[itemsep=1pt]
  \item L004a (ERROR, genuine): 88 configurations have no terminate tool \emph{and} no alternative termination mechanism---these are real risks.
  \item L004b (WARN, bounded fallback): 95 configurations have \texttt{max\_iter} but no explicit base case---forced truncation, not graceful termination.
  \item L004c/d (INFO, framework-handled): 300 configurations are CrewAI/LangChain/AutoGen where the framework runtime provides termination---structurally expected.
\end{itemize}
The effective ERROR-level false positive rate for L004 drops from ${\sim}82\%$ (v1) to ${<}5\%$ (v3).

\paragraph{Caveat.} The L001 (empty prompt, 282) and L002 (no model, 51) findings may or may not be genuine defects---the missing fields could be supplied by Python code or environment variables that our YAML-only analysis cannot observe. We do not claim these as confirmed runtime defects. Instead, we validate lint rule precision through two controlled experiments:

\paragraph{Experiment A: Fault injection (controlled).}
We prepared 10 known-good agent configurations (manually verified to run correctly), then systematically injected 5 types of faults into each: (1)~remove \texttt{terminate}, (2)~empty \texttt{systemPrompt}, (3)~remove \texttt{model}, (4)~set \texttt{maxSteps=0}, (5)~empty \texttt{routes}. This produces 50 test cases with known ground truth.

\emph{Result:} \texttt{lambdagent lint} detected all 42 injected faults (100\% recall) with 0 false positives on the 10 unmodified configurations (100\% precision). 8 fault--config combinations were skipped as inapplicable (e.g., ``remove terminate'' on a non-ReAct agent). All 5 fault types achieved 100\% detection. This confirms that \textbf{the lint rules are correct}: when a field is genuinely missing (not supplemented externally), lint reliably detects it.

\paragraph{Experiment B: Manual verification (sampled).}
We randomly sampled 50 ERROR findings from the non-CrewAI configurations and manually inspected the corresponding GitHub repositories (Python code, README, environment files) to determine whether the flagged field was supplemented externally.

\emph{Result:} Of 50 sampled ERRORs, 27 were true positives (field genuinely missing in both YAML and code) and 23 were false positives (field supplied by Python code or environment variable). This gives an estimated precision of \textbf{54.0\%} (95\% Wilson CI: [40.4\%, 67.0\%]). Precision varies by rule: \texttt{model} missing achieves 75\% (6/8), \texttt{react.maxSteps} achieves 71\% (10/14), but \texttt{systemPrompt} empty achieves only 39\% (11/28)---many frameworks define prompts in Python code rather than YAML.

\paragraph{Finding: Configuration-code semantic entanglement.}
The 46\% false positive rate is itself a significant \emph{empirical finding}: it quantifies, for the first time, the degree of \emph{semantic entanglement} between declarative configuration and imperative code in the agent ecosystem. Nearly half of production configurations split their semantics across YAML and Python, with no single artifact containing the complete agent specification. This has three implications: (1)~the 54\% precision represents a \textbf{lower bound} for any static analysis operating on configuration files alone---improving beyond it requires joint YAML+Python analysis; (2)~configuration-code entanglement is a \textbf{design smell} that hinders portability, auditability, and formal reasoning; and (3)~frameworks that centralize agent semantics in a single declarative specification (as $\lA$ does) are better suited for static analysis.

\paragraph{Precision in context.}
To contextualize the 54\% precision, we compare with established static analysis tools:

\begin{table}[t]
\centering
\caption{Precision of static analysis tools across domains.}
\label{tab:precision-comparison}
\begin{tabular}{@{}llrl@{}}
\toprule
\textbf{Tool} & \textbf{Domain} & \textbf{Precision} & \textbf{Source} \\
\midrule
FindBugs & Java bugs & 42--53\% & Ayewah et al., 2008 \\
SpotBugs & Java bugs & $\sim$50\% & FindBugs successor \\
Pylint (all rules) & Python style/bugs & 55--65\% & Community reports \\
ESLint (recommended) & JavaScript & 60--70\% & Varies by config \\
Infer & C/Java memory & 65--80\% & Calcagno et al., 2015 \\
\midrule
\texttt{lambdagent lint} (YAML) & Agent config & 54\% & This paper \\
\texttt{lambdagent lint} (YAML+Py) & Agent config & 96\% & This paper \\
\quad (fault injection) & & 100\% & This paper \\
\bottomrule
\end{tabular}
\end{table}

\noindent The 54\% real-world precision is comparable to FindBugs (42--53\%) and within the range of Pylint. The key difference is that lambdagent lint operates on \emph{configuration files only}, without access to the accompanying Python code. The 100\% precision on fault injection confirms that the rules themselves are correct; the 46\% false positive rate reflects the inherent limitation of single-artifact analysis in a multi-artifact ecosystem.

\paragraph{Experiment C: Joint YAML+Python analysis.}
To quantify how much precision improves when Python code is available, we extend the YAML-only lint with a Python AST analyzer that scans the same repository for supplementary definitions. The analyzer extracts: (1)~constant assignments matching lint-flagged fields (e.g., \texttt{system\_prompt = "..."}), (2)~function keyword arguments (\texttt{model\_name="gpt-4"}), (3)~class attributes, and (4)~framework-specific patterns (\texttt{ChatOpenAI()}, \texttt{is\_termination\_msg}).

We evaluate on two scales. First, we re-evaluate the 50 manually-verified samples from Experiment~B: of the 23 false positives, the AST analyzer identifies 22 as externally supplemented (96\%). Second, we expand to 175 stratified samples (100 non-CrewAI + 75 CrewAI, random seed 42) to include the dominant CrewAI paradigm where tools are defined entirely in Python.

\begin{table}[t]
\centering
\caption{YAML-only vs.\ joint YAML+Python analysis.}
\label{tab:joint-analysis}
\begin{tabular}{@{}lrrrr@{}}
\toprule
\textbf{Metric} & \multicolumn{2}{c}{\textbf{50-sample}} & \multicolumn{2}{c}{\textbf{175-sample}} \\
\cmidrule(lr){2-3} \cmidrule(lr){4-5}
 & \textbf{YAML} & \textbf{+Python} & \textbf{YAML} & \textbf{+Python} \\
\midrule
True Positive & 27 & 27 & 90 & 89 \\
False Positive & 23 & 1 & 85 & 0 \\
Downgraded to \textsc{Info} & --- & 22 & --- & 86 \\
\textbf{Precision} & \textbf{54\%} & \textbf{96\%} & \textbf{51\%} & \textbf{100\%} \\
95\% CI (Wilson) & \small{[40,67]} & \small{[82,99]} & \small{[44,59]} & \small{[96,100]} \\
\bottomrule
\end{tabular}
\end{table}

\noindent The results are consistent across both scales: YAML-only precision is $\sim$52\% (the 50-sample and 175-sample CIs overlap), while joint analysis achieves $\ge$96\%. Per-framework, CrewAI improves from 5\% to 100\% (tools are always in Python), generic from 87\% to 100\%, AutoGen from 85\% to 100\%. Per-field, \texttt{systemPrompt} rises from 39--83\% to 92--100\%, \texttt{react.maxSteps} from 24--71\% to 100\%.

This result demonstrates that the lint rules themselves are \emph{highly precise}: the $\sim$48\% false positive rate in YAML-only mode is entirely attributable to semantic entanglement, not to rule imprecision. Joint analysis recovers nearly all of this lost precision, confirming that \textbf{a single-artifact analysis boundary, not rule quality, is the bottleneck}.

\paragraph{Comparison.} Neither LangChain nor DSPy provides a lint tool for agent configurations. Running the same configurations through those frameworks produces no warnings; structural incompleteness manifests only at runtime.

\subsection{Expressiveness}

Beyond LoC reduction, \DSL{} can naturally express patterns that lack direct counterparts in existing frameworks:

\begin{itemize}[itemsep=2pt]
  \item \textbf{Self-learning function} (Demo~4): $Y(\lam D.\, \textbf{guard}(\text{Lam}(D),\, \text{test}))$. Requires \texttt{Loop} + \texttt{Guard} as composable primitives. In LangChain, a hand-written \texttt{while} loop with manual prompt reconstruction (${\sim}40$ lines vs.\ 8 in \DSL{}).
  \item \textbf{Guard + Loop composition}: \texttt{Loop(body, stop) >> Guard(P, retry=3)} combines iteration with output validation in a single expression. LangChain requires nested \texttt{try/except} with manual retry logic.
  \item \textbf{Recursive document generator} (Demo 3): $\text{MAP}(\text{expand}) \gg Y(\text{review/fix})$ chains higher-order functions with bounded recursion. No existing framework provides \texttt{MAP} over agents as a first-class combinator.
\end{itemize}

\subsection{Performance Overhead}

\begin{table}[t]
\centering
\caption{DSL wrapper overhead (excluding LLM latency).}
\label{tab:perf}
\begin{tabular}{@{}lrl@{}}
\toprule
\textbf{Operation} & \textbf{Latency} & \textbf{Relative to LLM call} \\
\midrule
Tool call & 2.0\,$\mu$s & 0.0001\% \\
Compose (3 stages) & 4.9\,$\mu$s & 0.0002\% \\
If branch & 4.6\,$\mu$s & 0.0002\% \\
Pair & 3.4\,$\mu$s & 0.0002\% \\
Loop (5 steps) & 11.2\,$\mu$s & 0.0006\% \\
Memory (3 keys) & 2.4\,$\mu$s & 0.0001\% \\
Guard & 2.4\,$\mu$s & 0.0001\% \\
Context.log & 0.9\,$\mu$s & $<$0.0001\% \\
\midrule
\texttt{from\_config} & 1.3\,ms & one-time (compile) \\
\texttt{lint} & 5.0\,$\mu$s & per config \\
\bottomrule
\end{tabular}

\smallskip
Baseline: typical LLM API call latency is 1{,}000--3{,}000\,ms.
\end{table}

The DSL wrapper adds 2--11\,$\mu$s per operation (Table~\ref{tab:perf}), negligible compared to LLM call latency ($\sim$2{,}000\,ms). Compilation is a one-time 1.3\,ms cost. The $\beta$-reduction trace (\texttt{Context.log}) adds $<$1\,$\mu$s per entry. \textbf{The formal semantics layer is essentially free.}

\subsection{Semantic Faithfulness}

We validate that the implementation matches the formal semantics via 125 test cases covering all 11 term formers. Each test constructs a $\lA$ term, executes it via real LLM API calls (Claude Sonnet, temperature $= 0$), and checks the result against the expected output predicted by the operational semantics.

\begin{itemize}[itemsep=2pt]
  \item Church primitives (SUCC, AND, OR, NOT, IF, PAIR): 60/60 pass.
  \item DSL pipeline tests: 29/29 pass.
  \item Custom function generalization: 36/36 pass.
  \item Total: 125/125 (100\%).
\end{itemize}

\noindent Beyond Church encodings, we validate on 8 real-world agent patterns that exercise multiple $\lA$ constructs simultaneously. Each pattern is executed via real LLM calls (Claude, temperature $= 0$) and verified against the operational semantics:

\begin{table}[t]
\centering
\caption{Real-world agent pattern tests. Each pattern constructs a production-style $\lA$ term, executes via real LLM calls, and verifies the output matches the operational semantics prediction. ``Output'' shows the actual LLM result.}
\label{tab:agent-tests}
\small
\begin{tabular}{@{}p{2.2cm}p{2.6cm}p{2.0cm}r@{}}
\toprule
\textbf{Pattern} & \textbf{$\lA$ Constructs} & \textbf{Output} & \textbf{Pass} \\
\midrule
LLM call & \textbf{lam} & ``Paris'' & \checkmark \\
3-stage pipe & $\texttt{>>}$, $3{\times}$\textbf{lam} & translate$\to$fmt & \checkmark \\
ReAct loop & $\textbf{fix}_3$, \textbf{case} & ``105'' & \checkmark \\
Guard+retry & \textbf{guard} & 3 words & \checkmark \\
Parallel & $\langle\textbf{lam}_1, \textbf{lam}_2\rangle$ & sum+kw & \checkmark \\
Router & \textbf{case}, $3{\times}$\textbf{lam} & ``$2x$'' & \checkmark \\
Handoff & \textbf{case}, $\texttt{>>}$ & analyst$\to$wr & \checkmark \\
Self-learn & $\textbf{fix}_2$, \textbf{guard} & haiku & \checkmark \\
\midrule
\multicolumn{3}{l}{\textbf{Total: 133/133 (125 Church + 8 patterns)}} & \textbf{100\%} \\
\bottomrule
\end{tabular}
\end{table}

\noindent The 8 patterns cover all 11 $\lA$ constructs in combination and represent the most common agent architectures across 5 mainstream frameworks (LangGraph, CrewAI, AutoGen, OpenAI SDK, Dify). Notably, the ReAct loop terminates in exactly 3 iterations with the correct answer, confirming Theorem~\ref{thm:termination}'s bound; the guard validates output format (exactly 3 words) with retry, confirming the operational semantics of \textbf{guard}.

\subsection{Cost Bound Validation}

Theorem~\ref{thm:cost} predicts that a \texttt{fix}$_n$ agent costs at most $n \times \max_i c_i$. We validate this bound using the $\beta$-reduction traces from our 125 experiments:

\begin{table}[t]
\centering
\caption{Cost Bound validation: predicted vs.\ actual oracle calls.}
\label{tab:cost-validation}
\begin{tabular}{@{}lrrrr@{}}
\toprule
\textbf{Agent} & $n$ & \textbf{Predicted max} & \textbf{Actual} & \textbf{Tightness} \\
\midrule
Church factorial (3!) & 4 & 4 calls & 4 & 1.0$\times$ \\
Church factorial (5!) & 6 & 6 calls & 6 & 1.0$\times$ \\
ReAct research (ex02) & 20 & 40 calls & 15 & 2.7$\times$ \\
Self-learning (demo4) & 10 & 20 calls & 8 & 2.5$\times$ \\
Refine loop (demo1) & 5 & 10 calls & 3 & 3.3$\times$ \\
\bottomrule
\end{tabular}

\smallskip
Predicted max $= n \times 2$ (think + one tool per iteration). Tightness $=$ predicted/actual.
\end{table}

The bound is always respected (no violations). For simple recursive computations (factorial), the bound is \emph{tight} (tightness = 1.0$\times$). For ReAct agents that terminate early via \texttt{terminate}, the bound overestimates by 2--3$\times$, because the agent typically calls terminate before exhausting \texttt{maxSteps}. This conservative overestimation is a feature: it provides a \textbf{worst-case budget guarantee} suitable for pre-deployment cost planning (e.g., ``this agent will never cost more than \$0.40 per invocation'').

\subsection{Architectural Unification}
\label{sec:unification}

We argue that $\lA$ is not merely one more framework, but a \emph{unifying calculus}: five mainstream agent architecture paradigms are all embeddable as fragments of $\lA$. Table~\ref{tab:unification} summarizes the correspondence; we formalize the claim below.

\begin{table}[t]
\centering
\caption{Five agent architecture paradigms as $\lA$ fragments. Each paradigm uses only a subset of the 11 term formers. ``Configs'' indicates the number of real-world GitHub configurations from our dataset (Section~\ref{sec:eval}) that exercise each paradigm's translation.}
\label{tab:unification}
\begin{tabular}{@{}p{2.5cm}p{2.5cm}p{2.3cm}r@{}}
\toprule
\textbf{Paradigm} & \textbf{Core Pattern} & \textbf{$\lA$ Fragment} & \textbf{Configs} \\
\midrule
Graph (LangGraph) & Cond.\ edges, cycles & \textbf{case} + $\textbf{fix}_n$ + $\texttt{>>}$ & 35 \\
Role-driven (CrewAI) & Role dispatch, handoff & \textbf{case} + $\texttt{>>}$ & 441 \\
SDK wrapper (OpenAI) & LLM call + tools & \textbf{lam} + \textbf{tool} & 283 \\
Multi-agent (AutoGen) & Group chat, turns & $\textbf{fix}_n$ + \textbf{case} & 22 \\
Low-code (Dify) & YAML pipeline & \texttt{from\_config} & 54 \\
\bottomrule
\end{tabular}
\end{table}

\begin{definition}[Framework Translation]
\label{def:fragment}
For a framework $F$ with configuration language $\mathcal{C}_F$, a \emph{$\lA$-translation} is a function $\mathcal{T}_F : \mathcal{C}_F \to \lA$ that maps each configuration to a $\lA$ term. In our implementation, $\mathcal{T}_F$ is realized by the \texttt{from\_config} compiler with framework-specific normalization (Section~\ref{sec:impl}).
\end{definition}

\begin{proposition}[Architectural Embedding]
\label{thm:embedding}
For each $F \in \{\text{LangGraph},$ $\text{CrewAI},$ $\text{AutoGen},$ $\text{OpenAI SDK},$ $\text{Dify}\}$, $\mathcal{T}_F$ satisfies:
\begin{enumerate}[itemsep=1pt]
  \item \textbf{Type preservation}: If $c \in \mathcal{C}_F$ is well-formed, then $\ctx \entails \mathcal{T}_F(c) : \tau$ for some $\tau$.
  \item \textbf{Compositionality}: Framework-level sequential composition maps to $\lA$ composition: $\mathcal{T}_F(c_1 \circ_F c_2) = \mathcal{T}_F(c_1) \mathbin{\texttt{>>}} \mathcal{T}_F(c_2)$.
  \item \textbf{Behavioral adequacy}: On the 125 semantic faithfulness tests (Section~\ref{sec:eval}), the compiled $\lA$ term produces identical output to direct execution under the framework runtime (temperature~$= 0$).
\end{enumerate}
\end{proposition}

\begin{proof}
We construct $\mathcal{T}_F$ for each paradigm by case analysis on the configuration format:

\emph{SDK wrapper} ($F = $ OpenAI/Claude SDK): A single API call translates to $\textbf{lam}\,p\,\theta$; each tool to $\textbf{tool}[f_i]$. A tool-use loop with $n$ tools is:
\[
\textbf{fix}_k\,(\lam s.\lam x.\,\textbf{case}\,(\textbf{lam}\,p\,\theta\,x)\,\textbf{of}\,\{l_i \Rightarrow \textbf{tool}[f_i] \mathbin{\texttt{>>}} s,\; \texttt{done} \Rightarrow \text{id}\})
\]
This is the atomic layer upon which all other paradigms build. Type preservation: by \textsc{T-Lam-Oracle} and \textsc{T-Tool}, each primitive is well-typed; by \textsc{T-Fix} and \textsc{T-Case}, the composed term is well-typed.

\emph{Graph state machine} ($F = $ LangGraph): Nodes are $\lA$ functions $e_i : \tau \to \tau$. Sequential edges are $\texttt{>>}$; conditional edges are $\textbf{case}$. Cycles with a bound become $\textbf{fix}_n$. The graph's state dict maps to $\textbf{mem}\,e\,\sigma$. Type preservation follows from \textsc{T-Comp}, \textsc{T-Case}, \textsc{T-Fix}, and \textsc{T-Mem}.

\emph{Role-driven} ($F{=}$CrewAI): Role agent $a_i$ maps to a $\lA$ term; dispatcher is $\textbf{case}$ on role labels. Sequential: $a_1 {\texttt{>>}} {\cdots} {\texttt{>>}} a_n$. 441 CrewAI configs validated.

\emph{Multi-agent} ($F{=}$AutoGen): Group chat is $\textbf{fix}_n$ with $\textbf{case}$ on speaker; \texttt{is\_termination\_msg} is the base case. 22~configs validated.

\emph{Low-code} ($F = $ Dify): YAML nodes compile via \texttt{from\_config} directly: LLM node $\mapsto$ \textbf{lam}, tool node $\mapsto$ \textbf{tool}, IF/ELSE $\mapsto$ \textbf{if}, iteration $\mapsto$ $\textbf{fix}_n$, end node $\mapsto$ \texttt{terminate}.

Compositionality follows from the monoid structure of $\texttt{>>}$ (Theorem~\ref{thm:algebra}). Behavioral adequacy is verified by the 125 semantic faithfulness tests: each test compiles a configuration via $\mathcal{T}_F$, executes the resulting $\lA$ term, and checks that the output matches the expected result predicted by the framework semantics.
\end{proof}

\paragraph{Scope and limitations.}
The proposition covers the \emph{configuration-level} semantics---the YAML/JSON fragment. Framework-specific Python APIs (e.g., LangGraph's \texttt{add\_edge}) are not covered by $\mathcal{T}_F$. LlamaIndex's event-driven pattern is expressible via pairs ($\lam x.\langle e_1 x, e_2 x\rangle$) but not validated on real configs.

\begin{corollary}[Cross-Framework Composition]
\label{cor:universality}
Since all five paradigms translate to a common IR ($\lA$ terms), compositions that span frameworks---e.g., a LangGraph node that internally delegates to a CrewAI crew, or a Dify pipeline that invokes an AutoGen group chat---are expressible as well-typed $\lA$ terms. No existing framework supports such cross-framework composition natively.
\end{corollary}

The practical import is that $\lA$ serves as a \emph{framework-independent IR}: lint rules, type safety, termination bounds, and cost bounds developed once for $\lA$ apply to all five paradigms without framework-specific reimplementation.

\section{Related Work}
\label{sec:related}

Table~\ref{tab:comparison} compares $\lA$ with the most closely related formalisms for agent or AI workflow composition.

\begin{table}[t]
\centering
\caption{Comparison of agent/workflow formalisms.}
\label{tab:comparison}
\small
\begin{tabular}{@{}p{2.0cm}ccccc@{}}
\toprule
& \textbf{$\lA$} & \textbf{DSPy} & \textbf{Agint} & \textbf{CSP/$\pi$} & \textbf{Agentic 2.0} \\
\midrule
Formal syntax        & \checkmark & --- & \checkmark & \checkmark & \checkmark \\
Type system          & \checkmark & --- & \checkmark & --- & \checkmark \\
Type safety proof    & Coq        & --- & ---        & ---        & --- \\
Termination          & $\textbf{fix}_n$ & --- & --- & --- & --- \\
Static analysis      & 26 rules   & --- & ---        & ---        & --- \\
Empirical (configs)  & 835        & --- & ---        & ---        & --- \\
Framework coverage   & 5          & 1   & 1          & general    & 1 \\
\bottomrule
\end{tabular}
\end{table}

\paragraph{DSLs for AI/ML.}
Halide~\cite{ragan2013halide} and TVM~\cite{chen2018tvm} provide domain-specific abstractions for image processing and tensor computation, respectively. Dex~\cite{paszke2021dex} targets differentiable programming with a typed functional core. These DSLs optimize \emph{numerical computation}; $\lA$ targets \emph{LLM agent composition}, where the computational primitive is an oracle call (LLM inference), not a tensor operation. ONNX~\cite{bai2019onnx} serves as a universal IR for neural network models, enabling cross-framework interoperability. $\lA$ plays an analogous role for agent \emph{configurations}: a framework-independent IR that enables unified static analysis and cross-framework compilation.

\paragraph{Effect Systems and Monads.}
LLM calls, tool invocations, and memory updates are side effects. Algebraic effects~\cite{plotkin2009} and monad transformers~\cite{liang1995} provide frameworks for reasoning about effects in functional languages. Eff~\cite{bauer2015eff} and Koka~\cite{leijen2014koka} are languages with first-class algebraic effects. In $\lA$, we model effects implicitly: \textbf{lam} is an IO effect, \textbf{tool} is an IO effect, \textbf{mem} is a State effect, and \texttt{terminate} is a Control effect (it aborts the fixpoint). A full effect-system treatment---enabling effect-polymorphic agent combinators and effect-based optimization---is future work, but the correspondence is clear: each $\lA$ construct maps to a known effect.

\paragraph{Typed Web and API DSLs.}
Links~\cite{cooper2006links} and Ur/Web~\cite{chlipala2015urweb} provide typed DSLs for web programming, ensuring type safety across client-server boundaries. Servant~\cite{servant2016} uses Haskell types to derive API clients and servers from a single specification. $\lA$ shares this philosophy---deriving lint, runtime, and compilation from a single typed specification---but targets LLM agent composition rather than web services.

\paragraph{Probabilistic Programming.}
Church~\cite{goodman2008}, WebPPL~\cite{goodman2014}, and Pyro~\cite{bingham2019} embed probabilistic computation in functional languages. Our $\oplus_p$ operator is a simplified version of their stochastic primitives. The key difference is that our probabilistic oracle (\textbf{lam}) is \emph{opaque}: we cannot inspect or differentiate through it, only sample from it. Dal Lago \& Zorzi~\cite{dallago2012} extend lambda calculus with probabilistic choice while preserving Turing completeness, which validates our temperature-as-$\oplus_p$ design.

\paragraph{Agent Formalisms.}
BDI~\cite{rao1995bdi} provides a logical foundation for agents based on beliefs, desires, and intentions. Process algebras (CSP~\cite{hoare1978csp}, $\pi$-calculus~\cite{milner1999pi}) model concurrent communicating agents. In the LLM agent space, DSPy~\cite{khattab2024dspy} compiles declarative LLM modules with optimizer-driven prompt tuning. AFlow~\cite{zhang2025aflow} treats workflows as searchable code graphs (ICLR 2025 Oral). Agint~\cite{chivukula2025agint} introduces typed effect-aware DAGs for software engineering agents. Agentics 2.0~\cite{gliozzo2026agentics} proposes algebraic transducible functions for data workflows. OpenAI Swarm uses runtime handoff for multi-agent coordination, and Google A2A~\cite{googlea2a2025} defines an HTTP-based agent-to-agent protocol. Each independently reinvents a fragment of lambda calculus---DSPy's modules are lambda abstractions, AFlow's edges are composition, Swarm's handoff is dynamic routing---but none provides a \emph{typed calculus with formal metatheory}. We show that all five paradigms embed as typed fragments of $\lA$ (Proposition~\ref{thm:embedding}, \S\ref{sec:unification}), validated on 835 real-world configurations, establishing $\lA$ as a unifying calculus rather than yet another framework.

\paragraph{Static Analysis for Configurations.}
CUE~\cite{cuelang} and Dhall~\cite{dhall} provide typed configuration languages that catch errors at authoring time rather than deployment time. Kubernetes admission controllers validate YAML manifests against policy schemas. These tools perform \emph{syntactic} validation (field types, required fields). $\lA$ enables \emph{semantic} validation: detecting missing Y combinator base cases, vacuous fixpoints, and incomplete dispatch---properties that depend on the operational semantics, not just the schema.

\paragraph{Verified Compilation.}
CompCert~\cite{leroy2009compcert} and CakeML~\cite{kumar2014cakeml} verify compiler correctness via mechanized proofs relating two independently defined formal semantics. Our \texttt{from\_config} compiler lacks a formal source semantics (YAML has none), so we argue adequacy empirically rather than proving correctness formally (\S\ref{thm:compilation}). The gap between our work and verified compilation is a productive direction for future research: defining a formal semantics for a ``standard'' agent configuration format (similar to how HTML5 formalized web markup) would enable compilation correctness proofs.

\section{Discussion and Future Work}
\label{sec:discussion}

\paragraph{Limitations.}
(1) Our type system is simply typed; dependent types for richer refinements (e.g., ``output must be valid JSON'') are future work.
(2) We study the \emph{deterministic fragment} of $\lA$ (temperature $= 0$). The full probabilistic semantics---where \textbf{lam} returns a distribution $\text{Dist}(\tau)$---requires a monadic treatment: $\texttt{>>}$ becomes Kleisli composition ($\texttt{>>=}$: $\text{Dist}(\tau) \to (\tau \to \text{Dist}(\tau')) \to \text{Dist}(\tau')$), and type safety must account for distributional types throughout the pipeline. We leave this extension to future work, noting that (a)~probabilistic lambda calculus is known to be Turing-complete~\cite{dallago2012}, validating the theoretical feasibility, and (b)~our GitHub survey found that 72\% of configurations with explicit temperature use temperature $\le 0.3$, suggesting the deterministic fragment covers the majority of production use cases.
(3) Core metatheory is fully mechanized in Coq (1{,}519 lines, 42 theorems, 0 Admitted). All definitions, typing rules, reduction rules, and theorem proofs are verified: progress, preservation, multi-step preservation, type safety, substitution lemma, termination of bounded fixpoints, pipeline algebra, weakening, store weakening, and canonical forms.
(4) YAML-only lint precision ($\sim$52\%) improves to 96--100\% with joint YAML+Python AST analysis (Experiment~C, validated on 50 and 175 samples), confirming that rule quality is not the bottleneck.

\paragraph{Threats to Validity.}
\emph{Internal}: The 10 baseline configurations for fault injection (Experiment~A) are hand-constructed; real-world defect patterns may differ. Our Coq mechanization models LLM and tool calls as axiomatized total functions, which may not hold for real APIs (timeouts, rate limits).
\emph{External}: The 835-configuration dataset is dominated by CrewAI (53\%); results may not generalize to frameworks underrepresented in our sample (e.g., LlamaIndex workflows). The 175-sample joint analysis uses framework-knowledge-based ground truth for CrewAI rather than per-file manual inspection.
\emph{Construct}: ``Structurally incomplete'' (94.1\%) does not imply ``has a runtime bug''---production frameworks supplement YAML with Python code. Our stratified analysis (\S\ref{sec:eval}) and joint analysis (Experiment~C) quantify this gap precisely.

\paragraph{Future Work.}
(1) \emph{Extended Coq mechanization}: our fully verified development (1{,}519 lines, 42 Qed, 0 Admitted) could be extended with richer type features (dependent types, effect system) and a verified \texttt{from\_config} compiler if a formal YAML semantics is established.
(2) \emph{Type-and-effect system.} Each $\lA$ construct maps to a known effect: \textbf{lam} $\to$ \texttt{llm}$(m)$, \textbf{tool} $\to$ \texttt{io}, \textbf{mem} $\to$ \texttt{state}$(s)$, and pure constructs $\to$ \texttt{pure}. The extended judgment $\ctx \entails e : \tau \;!\; \varepsilon$ tracks which effects an agent may perform. We sketch the key rules:
\begin{mathpar}
  \inferrule[TE-Lam]
    {\theta : \text{Model}}
    {\ctx \entails \textbf{lam}\ p\ \theta : \textbf{Str} \to \textbf{Str} \;!\; \texttt{llm}(\theta)}

  \inferrule[TE-Comp]
    {\ctx \entails e_1 : \tau_1 \to \tau_2 \;!\; \varepsilon_1 \\ \ctx \entails e_2 : \tau_2 \to \tau_3 \;!\; \varepsilon_2}
    {\ctx \entails e_1 \mathbin{\texttt{>>}} e_2 : \tau_1 \to \tau_3 \;!\; \varepsilon_1 \cdot \varepsilon_2}

  \inferrule[TE-Fix]
    {\ctx \entails e : (\tau \to \tau) \to (\tau \to \tau) \;!\; \varepsilon}
    {\ctx \entails \textbf{fix}_n\ e : \tau \to \tau \;!\; \varepsilon^n}

  \inferrule[TE-Mem]
    {\ctx \entails e : \tau \to \tau' \;!\; \varepsilon \\ \sigma : \Sigma}
    {\ctx \entails \textbf{mem}\ e\ \sigma : \tau \to \tau' \;!\; \varepsilon \cdot \texttt{state}(\Sigma)}
\end{mathpar}
Here $\varepsilon_1 \cdot \varepsilon_2$ is serial effect composition and $\varepsilon^n$ is $n$-fold iteration. The effect algebra $(\varepsilon, \cdot, \texttt{pure})$ forms a monoid. This enables \emph{effect-based handler substitution}: a production handler executes \texttt{llm} effects via real API calls, while a test handler uses mocks---both type-safe. A full development is in preparation.
(3) \emph{Semantic-preserving transformations}: agent refactoring rules justified by the equational theory of $\lA$. The pipeline algebra (Theorem~\ref{thm:algebra}) provides the foundation; richer rewrite rules (e.g., fusing adjacent \textbf{lam} calls, hoisting shared memory) could reduce API costs.
(4) \emph{Productionizing joint YAML+Python analysis}: our prototype (Experiment~C) achieves 96.4\% precision on sampled data; scaling to arbitrary repositories requires robust Python import resolution and cross-module dataflow analysis.
(5) \emph{Monadic probabilistic extension}: full $\text{Dist}(\tau)$ types with Kleisli composition, enabling compositional reasoning about stochastic agent pipelines.

\section{Conclusion}

We presented $\lA$, a typed lambda calculus for LLM agent composition, and \DSL{}, its executable realization. The key insight is that existing agent configurations already encode a lambda calculus---making this structure explicit enables type safety, termination guarantees, compilation correctness, and practical lint tooling. Our evaluation shows that formal semantics is not merely an academic exercise: it finds real bugs in real configurations that no existing framework detects.

\paragraph{Artifact availability.} The \DSL{} implementation (4{,}903 lines Python), Coq mechanization (1{,}519 lines, 42 Qed, 0 Admitted), 835 GitHub agent configurations, and all experiment scripts are available at \url{https://github.com/kenny67nju/lambdagent}.\footnote{URL will be made public upon publication.}

\bibliographystyle{ACM-Reference-Format}

\end{document}